\title{Line-Broadcasting in Complete $k$-Trees}
\author{
R. Hollander Shabtai
\thanks{
School of Computer Sciences, Tel
Aviv University, Tel Aviv 69978, Israel and Afeka College  of Engineering, Tel-Aviv 69460, Israel }
\and Y. Roditty
\thanks{
School of Computer Sciences, Tel Aviv University, Tel Aviv 69978,
Israel and School of Computer Sciences, The Academic College of
Tel-Aviv-Yaffo, Tel-Aviv 61161, Israel. email: jr@post.tau.ac.il}}
\date{} 

\documentclass[a4paper,11pt]{article}

\newcommand{\tab}[1]{\hspace{.1\textwidth}\rlap{#1}}

\usepackage{latexsym}
\usepackage{graphicx}
\usepackage{algorithm}
\usepackage{algorithmic}
\usepackage{amsmath}
\usepackage{amssymb}
\usepackage{bbm}
\usepackage{amsthm}
\usepackage{amsfonts}
\usepackage{float}
\usepackage{amsmath,amssymb,graphics,graphicx,float,subfigure}
\newtheorem{lemma}{Lemma}[section]

\newtheorem{remark}{Remark}[section]

\begin{document}
\maketitle
\begin{abstract}

A {\textit line-broadcasting model} in a connected graph $G=(V,E)$, $|V|=n$, is a model in which
one vertex, called the {\it originator} of the broadcast holds a message that has to be transmitted
to all vertices of the graph through placement of a series of calls over the graph.
In this model, an informed vertex can transmit a message through a path of any length in a single time unit,
as long as two transmissions do not use the same edge at the same time.
Farley \cite{f} has shown that the process is completed within at most $\lceil \log_{2}n \rceil$ time units from any originator in a tree (and thus in any connected undirected graph) and that the cost of broadcasting one message from any vertex, i.e. the total number of edge used, is at most $(n-1) \lceil \log_{2}n \rceil$.

In this paper, we present lower and upper bounds for the cost to broadcast one message
in a complete $k-$tree, $k \ge 2$, from any vertex using the line-broadcasting model.
We prove that if $B(u)$ is the minimum cost to broadcast in a graph $G=(V,E)$ from a vertex $u \in V$ using the line-broadcasting model, then $(2-o(1))n \le B(u) \le (2+o(1))n$, where $u$ is any vertex in a complete $k$-tree. Furthermore, for certain conditions, $B(u) \le (2-o(1))n$.



\end{abstract}

\section{Introduction}
Broadcasting is the
process of message transmission in a communication network. The
communication network is modeled by a graph $G=(V,E)~,|V|=n$, where the set
of vertices $V$ represents the network members, and the set of edges
$E$ represents the communication links between given pairs of vertices.
We assume that $G$ is connected and undirected. We further assume that one vertex, called the\emph{ originator} of the graph, holds a message that has to be transmitted
to all vertices of the network through a placement of a series of calls over the
network. We call this a \emph{local broadcasting model}.

\emph{Line broadcasting}, sometimes called a \emph{wormwole} and \emph{cut-through} communication protocol, is a process in which a vertex can transmit the message
to any vertex in the graph through a path of any length in just one
time unit. The line-broadcasting model is applied in circuit-switched networks,
wormhole routing, optical networks, ATM switching and networks
supporting connected mode routing protocols.

In the line-broadcasting model, in a given time unit, two different calls cannot use the same edge; i.e., the paths used by two simultaneous calls must be edge-disjoint.
The \emph{cost} of a call is the number of edges used by the call, which is
the number of edges in the path between the call's transmitter and
receiver.

A {\it broadcasting scheme} specifies of which calls are
scheduled at each time unit (in any broadcasting model) and which paths are used in each call.

It is easily observed that the
lower bound on the number of time units needed to broadcast in a
graph $G$  is $\lceil{\log_{2} n}\rceil$. In the local broadcasting model, the possibility
to reach this lower bound depends on the graph topology. However, in the line-broadcasting model, Farley \cite{f} has proven that
the process of line-broadcasting can be completed within at most $\lceil{\log_{2} n}\rceil$ time units from any originator in a tree, and thus in any connected (undirected) graph.

Herein we measure the \textit{total time} and the \textit{cumulative cost}
of the broadcasting scheme. The total time of the broadcasting
scheme is equal to the number of time
units the broadcasting scheme needs in order to complete broadcasting.
The cumulative cost is the sum of the number of edges used by all calls
at each time unit cumulated on all time units.

 Broadcasting in communication networks
 has been investigated in the literature since the early 1950s
 (see surveys on broadcasting under various models
 and different topologies \cite{hhl},\cite{fl}).
 With the growing interest in parallel and telecommunication systems, a vast
literature has been devoted to specific group of communication
setups on specific network topologies.

The general broadcasting problem under the single-port local model
has been shown to be NP-complete; however, if the graph is a tree, the problem of
finding an optimal broadcasting scheme is polynomial
\cite{p,sch}.

Additional topologies that have been investigated include the complete graph,
torus graph, ring, grid, hypercube, shuffle-exchange and butterfly
graph, with a recent generalization of weighted trees in
\cite{ars1}.

 Analysis of broadcasting in grids was first investigated by
 Farley and Hedetniemi \cite{fh}. Van-Scoy and Brooks \cite{vs} extended their
 results to broadcasting of $m$ messages from
 a corner of a $2$- or $3$-dimensional grid.
 These results were extended to a $d$-dimensional grid
 by Roditty and Shoham \cite{rs}, who also developed
 an efficient algorithm for broadcasting from any
 originator in a $d$-dimensional grid.

Cohen, Fraigniaud, and Mitjana \cite{cfm} summarized current results and proposed new schemes for achieving minimum-time line
broadcasting in trees and in directed trees.
Averbuch, Hollander-Shabtai and Roditty \cite{ahr} introduced a line-broadcasting algorithm in a tree in the $k$-port model, where each vertex can transmit a message to at most $k$ vertices at each time unit. They also gave exact bounds for the cost of optimal $k$-port line-broadcast in stars and in complete trees.



Studies that investigate minimum-time single-port line-broadcasting schemes typically seek to identify means of minimizing the cumulative costs of
such schemes. The cost of Farley's scheme in \cite{f} is at most
$$(n-1)\lceil \log_2 n \rceil, \eqno (1)$$ where $n$ is the number of vertices in
the graph.

Kane and Peters \cite{kp} determined the value of the minimum cost
for a minimum-time line-broadcasting schemes in any cycle with $n$ vertices.
For $n=2^k$ they gave an exact value, while for other values of $n$ an
upper bound was presented. In each case the cost was about $\frac{1}{3}$ of Farley's upper bound, in (1).

Fujita and Farley \cite {ff} discussed minimum-time line
broadcasting in paths. The cost of their scheme was dependent on the
position of the originator in the path, and they specifically focused
on an originator that is a leaf of a path
or that is the father or grandfather of a leaf of a path. In any case, the
cost of the line-broadcasting scheme they introduced, was again
about $\frac{1}{3}$ of Farley's upper bound, in (1).

Averbuch, Roditty  and Shoham
\cite{ars2} developed efficient line-broadcasting algorithms in a
$d$-dimensional grid. These algorithms produce a linear cost as a function of
the number of vertices in the graph, with any vertex as the originator.

Averbuch, Gaber and Roditty \cite {agr1} studied line broadcasting in
complete binary trees. They provided a minimum-time line
broadcasting scheme originating from any vertex of the tree.
The lower and upper bounds they obtained, were, again,
$O(n)$ where, $n$ is the cardinality of the vertex set of the tree.


As far as we know there are no other known results concerning
the line-broadcasting model in which the cumulative cost is $O(n)$ and $n$ is the
size of the vertex set of the topology.

Thus, keeping in mind the result of Fujita and Farley \cite{ff}, concerning paths, on one hand, and that of Averbuch, Gaber and Roditty \cite{agr1} on the other hand, we pose the following problem:


\textbf{Problem:} Characterize the trees $T=(V,E)$, $|V|=n$, such that the cumulative cost using the line-broadcasting model from any vertex is linear in terms of n.

In this work, we provide some insight into the solution to this problem by extending the work of Averbuch, Gaber and Roditty \cite {agr1}. Specifically, we introduce lower and upper bounds for the cumulative cost
of a line-broadcasting process in a complete $k-$ tree , $k \ge 2$.
As we show , both bounds are linear in terms of $n$,
the cardinality of the vertex set of the  complete $k$-tree.

Let $B(u)$ denote the minimum cost to broadcast in a graph $G=(V,E)$ from a vertex $u \in V$ using the
line-broadcasting model.

Our main results are:

\textbf{Theorem 1:}
\label{t1}

Let $T=(V,E)$, $|V|=n$,
$n =\frac{k^{r+1}-1}{k-1}$, where, $k \ge 2 ~, r \ge 1 $,
are positive integers,
be a complete $k$-tree with height $r$.
Then,

$$ B(u) \ge (2-\frac{2(k-1)}{k^{2}})n-\frac {\lceil \log k \rceil}{k})-\frac{2}{k^{2}} + \frac{\lceil \log k \rceil}{k}-1.$$

For the upper bound we show:

\textbf{Theorem 2:}
\label{t2}

Let $T=(V,E)$, $|V|=n$,
$n =\frac{k^{r+1}-1}{k-1}$, where, $k \ge2 ~, r \ge 1 $,
are positive integers,
be a complete $k$-tree with height $r$.
Then, $B(u)$ is at most,

\begin{enumerate}

\item $(2 - \frac{\lceil \log_{2}(k+1) \rceil}{k})n-2+\frac{\lceil \log_{2}(k+1) \rceil}{k}$,  if  $r \lceil \log_{2}(k+1) \rceil \le \lceil \log_{2}n \rceil$.

\item $[2-\frac{(k-1)}{k^2}\lceil \log_{2}(k+1)\rceil + \frac{1}{k(k-1)}]n-2(r-1)+\frac{k}{(k-1)^2}+\frac{1}{k}-\frac{\lceil \log_{2}(k+1) \rceil}{k^2}$,  if $\lceil \log_{2}(n-k^{r}) \rceil + \lceil \log_{2}(k+1) \rceil \le  \lceil \log_{2}n \rceil$.

\item $(2 + { \frac{1}{k-1}})n + 2r \lceil \log_2 k^{r} \rceil - 2 \lceil \log_2 (k^{r}+1) \rceil  - 3r- { \frac{r+1}{k-1}}$,  otherwise.


\end{enumerate}




\begin{remark}
For $k=2$, the lower and upper bounds in \cite{agr1} are better than those derived from Theorems 1 and 2.
\end{remark}
The reminder of this paper is organized as follows: In section 2 we provide some definitions and notation that will be used in subsequent sections. In section 3 we prove Theorem 1, which provides a lower bound for the cost of line-broadcasting in complete $k$-trees. In section 4 we describe two procedures that perform partial line-broadcasting in complete $k$-trees. The two procedures will be used to develope our main algorithm, which is the main tool for the proof of Theorem 2. In section 5 we prove Theorem 2 by presenting our main algorithm. The algorithm, called LBCKT uses three algorithms for line-broadcasting in complete $k$-trees; these algorithms denoted Alg1, Alg2 and Alg3, are also presented in section 5. The algorithm LBCKT decides which of the three algorithms to use on the basis of the values of $k$ and $r$ in the expression $n=\frac{k^{r+1}-1}{k-1}$, with the goal of fulfilling the broadcast time constraint, $\lceil \log_{2}n \rceil$, and minimizing the total cost of the line-broadcasting process.

\section{Notation and Definitions}
Let $T=(V,E)$ be a rooted tree.

\begin{enumerate}
\item Denote by $root(T)$ the root of $T$.

\item A \textit{level} in a tree is a set of vertices that are at the same distance from $root(T)$. Let $T$ be a tree with $r+1$ levels.

Let $L_{j}$, $0 \le j \le r$, denote the set of vertices of level $j$ in a tree $T$.

\item A \textit{complete tree} is a tree, in which all leaves are at the same level.

\item The \textit{height} of a tree $T$, denoted by $r$, is the number of edges on a path from the $root(T)$ to the farthest leaf. A tree of height $r$ has $r+1$ levels labeled $0,...,r$, where $root(T)$ is at level $0$ and the farthest leaf is at level $r$.

\item Let $T=(V,E)$ be a rooted tree. For each $v \in V\setminus \{root(T)\}$, denote by $P(v)$, the parent of $v$, where $P(v) \in V$, $(P(v),v) \in E$ and $P(v)$ is on the path from $root(T)$ to $v$.

\item A \textit{$k-tree$}, $k \in N$, is a rooted tree, in which the number of children of each non-leaf vertex is exactly $k$. The degree of $root(T)$ is $k$, the degree of each non-leaf vertex excluding the $root(T)$ is $k+1$, and the degree of a leaf-vertex is $1$.

\item A \textit{complete $k$-tree}  is a rooted $k$-tree, in which each vertex has exactly $k$ children and all leaves are at the same level.

\item A \textit{complete tree} is a tree, in which all leaves are at the same level.

    \textbf{Observation:}
    \begin{enumerate}


    \item The number of vertices in a complete $k$-tree with height $r$ is $|V|=n=\frac{k^{r+1}-1}{k-1}$.

    \item The number of vertices in $L_{j}$, $0 \le j \le r$, in a complete $k$-tree with height $r$, is $k^{j}$.

    \end{enumerate}
\end{enumerate}

For other graph theory definitions refer to \cite{west}.

Throughout the paper, by $T=(V,E)$ we denote the complete $k$-tree.

\section{Proof of Theorem 1}

Observe first that
$$\lceil \log {n} \rceil \le r \log {k} +1 . \eqno(2)$$

Consider now the last $\lceil \log {k} \rceil$ time units of the line-broadcast. That is, from time unit $\lceil \log{n} \rceil - \lceil \log{k} \rceil+1$ to time unit $\lceil \log {n} \rceil$.

The number of internal vertices in the tree is $\frac{k^{r+1}-1}{k-1}-k^{r}-1=\frac{k^{r}-1}{k-1}-1=\frac{1}{k}\frac{k^{r+1}-k}{k-1}-1=\frac{1}{k}(n-1)-1 \le \frac{1}{k}(n-1)$.

Therefore, the number of calls done by internal vertices that transmit at the last $\lceil \log {k} \rceil$ time units is at most $\frac{n-1}{k} \lceil \log k \rceil$.

The number of vertices that are informed before the last $\lceil \log {k} \rceil $ time units is at most $ 2^{\lceil \log n \rceil -\lceil \log k \rceil}$, and therefore, using (2), the number of uninformed vertices is at least 
$$n-2^{\lceil \log{n} \rceil -\lceil \log k \rceil} \ge n-2\cdot 2^{(r-1) \log k } = n-2\frac{n(k-1)+1}{k^{2}}.$$ 


Each informed leaf, $v$, that transmits the message at the last $\lceil \log k \rceil$ time units has received the message via a call that used the edge $(v,P(v))$ and therefore, each call that $v$ initiates reuses the edge $(v,P(v))$.

The leaves transmit at the last $\lceil \log k \rceil$ time units to at least 

$n-2\frac{n(k-1)+1}{k^{2}} -\frac{n-1}{k}\lceil \log k \rceil$ vertices and therefore, with the additional obvious $n-1$ calls it follows that,

$$B(u) \ge n(2-\frac{2(k-1)}{k^{2}}-\frac {\lceil \log k \rceil}{k})-\frac{2}{k^{2}} + \frac{\lceil \log k \rceil}{k}-1.$$

Note that is case $u$ is a leaf the cost is decreased by 1.

\begin{remark}
The case where $k=2^{a}$ was dealt in \cite{gp}.
\end{remark}

\section{Procedures for partial line-broadcasting in a complete $k$-tree}
In this section we introduce two procedures that perform broadcasting to a part of a complete $k$-tree. The first procedure, $ToLevel$, broadcasts to all vertices in some level $L_{j}$, for a given $j$, $1 \le j \le r$. At the beginning of the procedure only the originator is informed. At the end of the procedure the originator and all vertices in $L_{j}$, and only they, are informed.

The second procedure, $FromLevel$, broadcasts from all vertices in some level $L_{j}$, for a given $j$, $1 \le j \le r$, to $L_{0},...,L_{j-1}$. At the beginning of the procedure only the originator and the vertices in $L_{j}$ are informed. At the end of the procedure all vertices in $L_{0},...,L_{j}$ are informed.
Both procedures shall be used to prove Theorem 2 in section 5.

The input for the following two procedures is:

 $T$ - a complete tree $k$-tree

 $k$ - the $k$-tree parameter (the number of children of each non leaf vertex)

 $r$ - the tree height

 $j$ - index of level $j$, $1 \le j \le r$ in the tree

 $u$ - the broadcast originator


\subsection{Line-broadcasting procedure to level $L_{j}$}

Label the vertices in $L_{j}$ from left to right:  $v_{1}^{j},...,v_{k^{j}}^{j}$.



\textbf{Def 1.} Define $j$ subsets of $L_{j}$, $S_{1},...,S_{j}$, such that for each $1 \le\ m \le j$, $$S_{m} = \{ v_{a_{m,i}}^{j}| a_{m,i} = 1+i k^{j-m}, 0 \le i \le k^{m}-1\}.$$

    As an immediate consequence of the definition we obtain,

    \begin{lemma}
    \label{l1}
    $S_{1} \subseteq S_{2} \subseteq ...  \subseteq S_{j}$.
    \end{lemma}

    \begin{lemma}
    \label{l2}
    $S_{j}=L_{j}$.
    \end{lemma}


\textbf{Def 2.}  Define $j$ sets of vertices $S_{1}',...,S_{j}'$, where $S'_{1}=S_{1}$ and for each $2 \le\ m \le j$, $S'_{m}=S_{m}-S_{m-1}$.

    \begin{lemma}
    \label{l3}
    $S'_{1} \cup S'_{2} \cup ... \cup S'_{j} = S_{j} = L_{j}$.
    \end{lemma}

    \begin{lemma}
    \label{l4}
    For each $p,q$, $1 \le p,q \le j$, $p \neq q$,
    $S'_{p} \cap S'_{q} = \phi$.
    \end{lemma}

    The proof of both lemmas follows immediately from Def 2 and from lemmas \ref{l1} and \ref{l2}.

    \textbf{Def 3.} For each vertex $v \in L_{j}$, $1 \le j \le r$, define $P$, the set of the ancestors of $v$, as the set of vertices that are on the path from $v$ to $root(T)$. $P=\{P(0),...,P(j-1)\}$, where for each $m$, $1 \le m \le j$, $P(m-1) \in L_{m-1}$. Note that $P(j-1)=P(v)$. Thus, the distance between $v$ and $P(m-1)$ is $j-m$. For each $P(i)$, $0 \le i \le j-1$, denote the subtree of $T$ rooted at $P(i)$ by $T_{P(i)}$.

We now present the procedure $ToLevel$. The procedure consist of $j$ rounds, where in each round $m$, $1 \le m \le j$, the originator and the informed vertices in $L_{j}$, which are the members of $S_{m-1}$, broadcast to the vertices in $S'_{m}$. Specifically, an informed vertex $v_{s}^{j} \in S_{m-1}$ broadcasts to a vertex $v_{t}^{j} \in S'_{m}$ that is in the sub tree of $P(m-1)$, where $s<t$. After round $m$ there are $|S_{m-1}|+|S'_{m}|$= $|S_{m}|$ informed vertices in $L_{j}$.

Note that the $m$ rounds of the procedure are not disjoint in time units. A given round $i$, $1 \le i \le m-1$, may end at time unit $t$, and round $m+1$ may begin at the same time unit.

\begin{algorithm}[h!!]
\caption{$ToLevel (T,k,r,j,u)$}
\label{alg:EDP2}
\begin{algorithmic}[1]
	\STATE 	for $m=1$ to $j$ do
	\STATE 	\tab{$u$ transmits to an uninformed vertex in $S'_{m}$}
    \STATE \tab{each informed vertex $v^{j}_{p}$ in $L_{j}$ transmits to an uninformed vertex $v^{j}_{q}$ }
    \STATE {\tab{in $S'_{m}$, such that $p < q < p+k^{j-m}$}}
	\end{algorithmic}
\end{algorithm}




\subsubsection{Correctness of the procedure $ToLevel$}
\begin{lemma}

\label{l5}
At the end of procedure $ToLevel$, all vertices in $L_{j}$ are informed.
\end{lemma}

\begin{proof}
\begin{enumerate}
\item The number of vertices that are informed at round $m$, $1 \le m \le j$ , is $|S'_{m}| = |S_{m}|-|S_{m-1}|$.
\item The number of informed vertices in $L_{j}$ after round $m$, $1 \le m \le j$, is $k^{m}$, and the informed vertices are the vertices in $S_{m}$.
\end{enumerate}

Since $S_{j}=L_{j}$, at the end of the execution of the procedure all vertices in $L_{j}$ are informed.
\end{proof}

\begin{lemma}
\label{l6}
The procedure $ToLevel$ fulfills the edge-disjoint constraint.
\end{lemma}

\begin{proof}
 Any two calls that are executed at round $m$, $1 \le m \le j$, are either calls from the originator to a vertex in $L_{j}$ or between two vertices in $L_{j}$.

Consider two calls that are executed at round $m$, $1 \le m \le j$, from $v^{j}_{p}$ to $v^{j}_{q}$ and from $v^{j}_{s}$ to $v^{j}_{t}$, where, $p < s$, $q < t$ (it is obvious that $p \neq q$ and $s \neq t$, since $p,s$ are informed and $q,t$ are not informed before the calls are executed). Then, from line 3 in $ToLevel$, $p < q < p+k^{j-m}$ and $s < t < s+k^{j-m}$. From the definition of $S'_{m}$, there exists a positive integer $i$ such that $p = q+ i k^{j-m}$, and therefore $p,q$ and $s,t$ are not the same sub-tree rooted at $L_{m-1}$, and therefore the paths that are used by the two calls are edge-disjoint.

If $u \not \in L_{j}$, consider two calls that are executed at round $m$, $1 \le m \le j$, from $u$ to $v^{j}_{q}$ and from $v^{j}_{s}$ to $v^{j}_{t}$, where, $q < t$ (it is obvious that $u \neq v^{j}_{q}$ and $s \neq t$). Again, the two calls pass through an ancestor of the two vertices in $L_{m-1}$, and according to the definition of $S'_{m}$, the vertices that receive the message are not in the same sub tree of their ancestor in $L_{m-1}$. Therefore, the paths that are used by the two calls are edge disjoint.

\end{proof}

\subsubsection{The total time of the procedure $ToLevel$}
If $u \in L_{j}$, then the number of informed vertices at the end of the procedure is $|L_{j}| = k^{j}$. Since each vertex that is informed at time unit $i$ is active at all time units, with the possible exception of the last time unit, the number of informed vertices is doubled at each time unit. Therefore, the number of time units needed to complete the broadcast to $L_{j}$ is $\lceil \log_{2}k^{j} \rceil$. If $u \not\in L_{j}$, the number of informed vertices at the end of the broadcast is $|L_{j}|+1=k^{j}+1$, and for the same reason, the number of time units needed to complete broadcasting in $L_{j}$ is $\lceil \log_{2}(k^{j}+1) \rceil$.

\subsubsection{The cost of the procedure $ToLevel$}
If $u \in L_{j}$, then each call initiated by $u$ passes through a path whose length is at most $2j$, from $u$ through some ancestor of $u$ to the receiver. Therefore, the cost of each call from $u$ to some vertex in $L_{j}$ is at most $2j$. The number of time units needed to complete broadcasting to all the $k^{j}$ vertices is $\lceil \log_{2}k^{j} \rceil$, and since $u$ is active at all time units, it is also the number of calls  initiated by $u$. Thus, the cost of all calls initiated by $u$ is at most $2j\lceil \log_{2}k^{j} \rceil$. If $u \not\in L_{j}$ the number of time units needed to broadcast to all of the $k^{j}$ vertices of $L_{j}$ is $\lceil \log_{2}(k^{j}+1) \rceil$. Again, since $u$ is active at all time units, in this case, the cost of each call from $u$ to some vertex in $L_{j}$, is at most $j$. Thus, the cost of the calls initiated by $u$ is at most $j\lceil \log_{2}(k^{j}+1) \rceil$.

 Observe that the number of informed vertices after round $m$, $1 \le m \le j$, is $|S_{m}| = k^{m}+1$, while the number of vertices that receive the information in round $m$ is $|S_{m}|-|S_{m-1}| = |S'_{m}|= k^{m}+1-(k^{m-1}+1) = k^{m}-k^{m-1}$. Therefore, the number of time units needed to complete round $m$, $1 \le m \le j$, is $\lceil \log_{2}(k^{m}+1)\rceil -\lceil \log_{2}(k^{m-1}+1)\rceil$. Thus, the number of calls initiated by $u$ at round $m$ is $\lceil \log_{2}(k^{m}+1)\rceil -\lceil \log_{2}(k^{m-1}+1)\rceil$, and the number of calls initiated by the informed vertices in $L_{j}$ at round $m$ is
  $$k^{m}-k^{m-1}-(\lceil \log_{2}(k^{m}+1)\rceil -\lceil \log_{2}(k^{m-1}+1)\rceil).$$

Since each informed vertex $v$ in $L_{j}$ broadcasts at round $m$, $1 \le m \le j$, to a vertex $w$ in the sub-tree of $P(m)$, the cost of each call initiated by a vertex in $L_{j}$ is $2(j-m+1)$, where $j-m+1$ edges are used to transmit to $P(m)$, and $j-m+1$ edges are used to transmit from $P(m)$ to $w$.

Thus, the cost of all calls initiated by the informed vertices in $L_{j}$ at round $m$ is $2(j-m+1)[k^{m}-k^{m-1}-(\lceil \log_{2}(k^{m}+1)\rceil -\lceil \log_{2}(k^{m-1}+1)\rceil)]$.

Therefore, the cost of all $m$ rounds is at most 

$$2j\lceil \log_{2}k^{j} \rceil +2\sum_{m=1}^{j}((j-m+1)[k^{m}-k^{m-1}-(\lceil \log_{2}(k^{m}+1)\rceil -\lceil \log_{2}(k^{m-1}+1)\rceil)]).\eqno(3)$$

In order to simplify the expression in (3) we made the following calculations:
\begin{enumerate}
\item $\sum_{m=1}^{j}(j+1)(\lceil \log_{2}(k^{m}+1)\rceil -\lceil \log_{2}(k^{m-1}+1)\rceil)=$

 $=(j+1)(\lceil \log_{2}(k^{j}+1)\rceil-1)$.
\item
$\sum_{m=1}^{j}m(k^{m}-k^{m-1})=\frac{k^{j+1}(kj-j-1)+k}{(k-1)^{2}}-\frac{k^{j}(kj-j-1)+1}{(k-1)^{2}}=$

$=\frac{k^{j}(kj-j-1)}{k-1}+\frac{1}{k-1}$.
\item $\sum_{m=1}^{j}(j+1)(k^{m}-k^{m-1})= (j+1)(k^{j}-1)$.

\item $\sum_{m=1}^{j} {m( \lceil \log_{2}(k^{m}+1)\rceil -\lceil \log_{2}(k^{m-1}+1)\rceil) } \le j  \lceil \log_{2}(k^{j}+1)\rceil-j$
\end{enumerate}

Thus, the cost of $ToLevel$ is at most




$$ \frac{2k(k^{j}-1)}{k-1} + 2j \lceil \log_{2} k^{j} \rceil - 2 \lceil \log_{2} (k^{j}+1) \rceil -2j+2 $$
 $$ \eqno(4)$$



\subsection{A line-broadcasting procedure from level $L_{j}$ to all levels $L_{0},...,L_{j-1}$}
In this section we introduce the procedure $FromLevel$, which describes line-broadcasting from the informed vertices in $L_{j}$ to all vertices in $L_{0},...,L_{j-1}$.

At the beginning of the procedure, the originator and the vertices in $L_{j}$ are informed. At the end of the procedure, all vertices in $L_{0},...,L_{j}$ are informed.

\textbf{Def 4.} Let $k$ be a positive integer. Define a sequence of positive integers $\{a_{n}\}_{n=1}^{\infty}$, where
$a_{n}=1+\frac{k^{n-1}-1}{k-1}$. For each $1 \le i \le j-1$, $1 \le t \le k^{j}$, we define a sequence $b_{i,t} = v_{a_{i}+(t-1)k^{i}}$, where $i$ is the sequence index and $t$ is the element index in sequence $i$ for $1 \le i \le k^{j-1}$.


\begin{algorithm}[h!!]
\caption{$FromLevel (T,k,r,j,u)$}
\label{alg:EDP2}
\begin{algorithmic}[1]
	\STATE 	{for $1 \le i \le j-1$, $1 \le t \le k^{j}$, $v_{b_{i,t}}$ transmits to its ancestor in $L_{j-i}$}
	\end{algorithmic}
\end{algorithm}		


Since all calls take place at the same time unit, the procedure takes exactly 1 unit of time.

\subsubsection{Correctness of the procedure $FromLevel$}
\begin{lemma}
\label{l7}
At end of the procedure $FromLevel$,  all vertices in $L_{0},...,L_{j}$ are informed.
\end{lemma}

\begin{proof}
 According to Def 4, if $p \neq q$, then $a_{p} \neq a_{q}$, and for each $1 \le i \le j-1$, the ancestors of $b_{p,t}$  and $b_{q,t}$ in $L_{j-i}$ are distinct. Therefore, all vertices in $L_{j-i}$ are informed, and thus all vertices in $L_{0},...,L_{j-1}$ are informed by $FromLevel$.
\end{proof}

\begin{lemma}
\label{l8}
The procedure $FromLevel$ fulfills the edge-disjoint constraint.
\end{lemma}

\begin{proof}
If $p \neq q$, then $a_{p} \neq a_{q}$, and for each $i$, $1 \le i \le j-1$, the ancestors of $b_{p,t}$  and $b_{q,t}$ in $L_{j-i}$ are distinct. Therefore, for each $i$, $1 \le i \le j-1$, any two calls to vertices in $L_{j-i}$ are edge disjoint.

For each $i$, $1 \le i \le j-1$, and for each $m,s$, $1 \le m,s \le k^{j}$, $m \neq s$, $b_{i,m} \neq b_{i,s}$, the ancestors of $b_{m,t}$ and $b_{s,t}$ is not in the same level and therefore, the paths from $b_{m,t}$ and $b_{s,t}$ to their ancestors are edge-disjoint.
\end{proof}

\subsubsection{The cost of the procedure $FromLevel$}
For each $i$, $1 \le i \le j$, the cost of each call is $i$, since the call is from a vertex in $L_{j}$ to a vertex in $L_{j-i}$.
Since $|L_{j-i}|=k^{j-i}$, it follows that the cost of all calls to $L_{j-i}$ is $ik^{j-i}$, and the cost of all calls to all $L_{1},...,L_{j}$ is:
 $\sum _{i=1}^{j} ik^{j-i}$.
Note that if $u=root(T)$, we have to subtract $j$, i.e., the cost of a call from a vertex in $L_{j}$ to $root(T)$, and thus the cost of the procedure is at most $$\sum _{i=1}^{j} ik^{j-i}-j = \frac{k^{j+1}-kj-k+j}{(k-1)^2}-j. \eqno(5)$$


\section{Proof of Theorem 2}

In this section, we present our main algorithm, LBCKT, for carrying out broadcasting in a complete $k$-tree. LBCKT uses three line-broadcasting algorithms denoted Alg1, Alg2 and Alg3, which are also presented in this section.


The input to all four algorithms is the tree $T=(V,E)$, where $|V|=n=\frac{k^{r+1}-1}{k-1}$; the tree parameters $k$ and $r$, and the originator $u \in V$.



LBCKT decides which of the three algorithms to use - Alg1, Alg2 or Alg3 on the basis of the values of $k$ and $r$. Specifically,LBCKT selects the algorithms that is expected to result in the broadcast scheme with the lowest cost.
We denote the cumulative cost of $Algi$, $1 \le i \le 3$, as $B_{Algi}(n)$, $1 \le i \le 3$.
We shall show that $B_{Alg1}(n) < B_{Alg2}(n) < B_{Alg3}(n)$, and that LBCKT cannot use Alg1 and Alg2 for all $k$ and $r$, because for some values of $k$ and $r$, Alg1 and Alg2 do not fulfill the $\lceil \log_{2} n \rceil$ time constraint. Therefore, LBCKT uses these algorithms only for values of $k$ and $r$ for which the time constraint is fulfilled.

\begin{algorithm}[h!!]
\caption{LBCKT(T,k,r,u)}
\label{alg:EDP2}
\begin{algorithmic}[1]
	\STATE 	{if $ r\lceil \log_{2}(k+1) \rceil \le \lceil \log_{2}n \rceil$ then}
\item \textit{\tab{$Alg1(T,k,r,u)$}}
	\STATE 	{else if $\lceil \log_{2}(n-k^{r})\rceil + \lceil \log_{2}(k+1)\rceil \le \lceil \log_{2}n\rceil$ then}
	\STATE {\tab{\tab{$Alg2(T,k,r,u)$}}}
	\STATE {\tab{else $Alg3(T,k,r,u)$}}

	\end{algorithmic}
\end{algorithm}


In the following sections 5.1, 5.2 and 5.3 we present the three algorithms, Alg1, Al2 and Alg3.

\subsection{Alg1}
The algorithm consists of $r$ rounds, where each round takes $\lceil \log_{2}(k+1) \rceil$ time units. At each time unit at round $j$, $1 \le j \le r$, each informed vertex $v \in L_{j-1}$ transmits to one of its uninformed children, and each informed vertex $w \in L_{j}$ transmits to one of its uninformed siblings, if there are any.









\begin{algorithm}[h!!]
\caption{Alg1(T,k,r,u)}
\label{alg:EDP2}
\begin{algorithmic}[1]

\STATE {if $u \in L_{0} \cup L_{1}$ or $k = 2^{p}$}
\STATE{\tab{for $1 \le i \le \lceil \log_{2}(k+1) \rceil$}}
\STATE{\tab{\tab{$u$ transmits to an uninformed $v \in L_{1}$}}}
\STATE{else}
\STATE{\tab{$u$ transmits to $root(T)$}}
\STATE{\tab{for $2 \le i \le \lceil \log_{2}(k+1) \rceil$}}
\STATE{\tab{\tab{$root(T)$ transmits to an uninformed $v \in L_{1}$}}}

\STATE{\tab{each informed $v \in L_{1}$ transmits to an uninformed sibling.}}

\STATE{for $2 \le j \le r$}
\STATE{\tab{for $1 \le i \le  \lceil \log_{2}(k+1) \rceil$}}
\STATE{\tab{\tab{each informed $v \in L_{j-1}$ transmits to an uninformed child.}}}
\STATE{\tab{\tab{each informed $v \in L_{j}$ transmits to an uninformed sibling}}}
\STATE{if $u \in L_{1}$ or $k = 2^{p}$}
\STATE{\tab{at time unit $\lceil \log_{2}n \rceil$ some $v \in L_{1}$ transmits to $root(T)$}}
	\end{algorithmic}
\end{algorithm}		




It is easily observed that each of the algorithm rounds takes
$\lceil \log_{2}(k+1) \rceil$ time units, and therefore the line-broadcasting in $T$ is completed within $r \lceil \log_{2}(k+1) \rceil$ time units.


\begin{remark}
Since $n=\frac{k^{r+1}-1}{k-1}$, it follows that for $k = 2^{p}-1$, where $p>2$ is an integer, and $\lceil \log_{2}(k^{r+1}-1)\rceil=\lceil \log_{2}k^{r+1} \rceil$ , Alg1 completes the line-broadcasting within $r \lceil \log_{2}(k+1) \rceil = \lceil \log_{2}n \rceil$ time units.

\end{remark}

\subsubsection{Correctness of Alg1}
\begin{lemma}
\label{l9}
At the end of the execution of the algorithm, all vertices in $T$ are informed.
\end{lemma}
\begin{proof}
First, consider $root(T)$. If $root(T)=u$, $root(T)$ is informed at the beginning of the broadcast. Otherwise, if $u \in L_{1}$ or $k = 2^{p}$, $root(T)$ is informed at the last time unit (lines 13,14 in Alg1). Otherwise, $root(T)$ is informed by $u$ at time unit 1 (line 5). Thus, at the end of the algorithm $root(T)$ is informed.

The vertices in $L_{1}$ are informed either by $u$ (lines 2-3 or 6-7), or by an informed sibling (line 8).

Consider now the vertices in $L_{2},...,L_{r}$. Each round $j$, $2 \le j \le r$, takes $\lceil \log_{2}(k+1) \rceil$ time units, which is the time needed to complete the broadcast from the parent in $L_{j-1}$ to its $k$ children in $L_{j}$. Therefore, each vertex in $L_{j}$ is informed either by its parent or by an informed sibling in $L_{j}$ (lines 9-12). Thus, after round $j$, $2 \le j \le r$, all vertices in $L_{j}$ are informed.

Thus, at the end of the algorithm all vertices in $L_{0},...,L_{r}$ are informed.

\end{proof}

\begin{lemma}
\label{l10}
The algorithm Alg1 fulfills the edge disjoint constraint.
\end{lemma}
\begin{proof}
At each round $j$, $1 \le j \le r$, a call may be either from a vertex in $L_{j-1}$ to one of its children or between two siblings in $L_{j}$. Consider two calls, $c_{1}$, $c_{2}$, that are executed at the same time unit, at round $j$. There are three possible cases:
\begin{enumerate}
\item $c_{1}$ is between a vertex $v \in L_{j-1}$ and a vertex $y \in L_{j}$, and $c_{2}$ is between a vertex $z \in L_{j-1}$ and a vertex $w \in L_{j}$. Thus, $c_{1}$ uses only the edge $(v,y)$, and $c_{2}$ uses only the edge $(z,w)$. Since $v \neq z$ and $T$ is a tree, then $y \neq w$, and thus $c_{1}$ and $c_{2}$ are edge-disjoint.

\item $c_{1}$ is from a vertex $v \in L_{j-1}$ to a vertex $y \in L_{j}$, and $c_{2}$ is from a vertex $z \in L_{j}$ to its sibling $w$. Thus, $c_{1}$ uses only the edge $(v,y)$, and $c_{2}$ uses two edges $(z,P(z))$ and $(P(z),w)$. Since $y \neq z$ and $y \neq w$, $c_{1}$ and $c_{2}$ are edge-disjoint.

\item $c_{1}$ is from a vertex $v \in L_{j}$ to its sibling $y$, and $c_{2}$ is from a vertex $z \in L_{j}$ to its sibling $w$.
    Thus, $c_{1}$ uses the edges $(v,P(v))$ and $(P(v),y)$, and $c_{2}$ uses two edges $(z,P(z))$ and $(P(z),w)$. Since $v \neq z$, $y \neq z$ and $y \neq w$, $c_{1}$ and $c_{2}$ are edge-disjoint.

\end{enumerate}
In addition, all calls from a vertex $v \in L_{j-1}$ to its child $x \in L_{j}$ use the edge $(v,x)$, which is the only edge on the path between them. All calls from a vertex $w \in L_{j}$, whose parent is $v$, to its sibling $z$, $z \neq x$, use the edges $(w,v)$ and $(v,z)$, and no other calls use these edges.
\end{proof}

\subsubsection{The cost of Alg1}
 We calculate the cost of each round.
 For each $j$, $1 \le j \le r$:
 \begin{enumerate}
 \item Each informed vertex $v \in L_{j-1}$ broadcasts to $\lceil \log_{2}(k+1) \rceil$ of its children, where the cost of each call is exactly 1. Thus, the cost of the calls initiated by $v$ is $\lceil \log_{2}(k+1) \rceil$. Since $|L_{j-1}| = k^{j-1}$, the cost of all calls from all vertices in $L_{j-1}$ to their children is $$k^{j-1} \lceil \log_{2}(k+1) \rceil.$$

 \item For each vertex $v \in L_{j-1}$, the number of children that transmit to their siblings is $k-\lceil \log_{2}(k+1) \rceil$, and the cost of each call is exactly 2. Since $|L_{j-1}| = k^{j-1}$, then the cost of all such calls is, $2k^{j-1}(k-\lceil \log_{2}(k+1) \rceil)$ at each round.
 \end{enumerate}
Thus, for each $j$, $1 \le j \le r$, the cost of round $j$ is $$k^{j-1} \lceil \log_{2}(k+1) \rceil + k^{j-1}(2(k-\lceil \log_{2}(k+1) \rceil)) = k^{j-1}(2k-\lceil \log_{2}(k+1) \rceil),$$

and the cost of all $r$ rounds is at most,
$$(2k-\lceil \log_{2}(k+1) \rceil)\sum_{j=1}^{r} k^{j-1} =(2-\frac{\lceil \log_{2}(k+1) \rceil}{k})n +\frac{\lceil \log_{2}(k+1) \rceil}{k}-2.$$

This proves (1) of theorem 2.

\subsection{Alg2}
This algorithm is based upon the procedures $ToLevel$ and $FromLevel$ and performs additional calls to the tree leaves.
The algorithm consist of two rounds.
At the first round, the originator, $u$, broadcasts to all the vertices in $L_{r-1}$, and in the second round, $u$ and the vertices in $L_{r-1}$ broadcast to all the remeining tree vertices of $T$, namely, to the vertices in $L_{i}$, $1 \le i \le r-2$, and to $L_{r}$, the set of $T$ leaves.

\begin{algorithm}[h!!]
\caption{Alg2(T,k,r,u)}
\label{alg:EDP2}
\begin{algorithmic}[1]
	\STATE 	{ToLevel   $(T,k,r,r-1,u)$}
	\STATE 	{FromLevel $(T,k,r,r-1,u)$ and at the same time broadcast to the tree leaves (on each star rooted in a vertex in $L_{r-1}$)}
	\end{algorithmic}
\end{algorithm}


\subsubsection{Correctness of Alg2}
\begin{lemma}
\label{l11}
At the end of the execution of the algorithm all vertices in $T$ are informed.
\end{lemma}

\begin{proof}
From lemma \ref{l5}, it follows that after the execution of $ToLevel(T,k,r,r-1,u)$, $u$ and all vertices in $L_{r-1}$ are informed. From lemma \ref{l7} it follows that after the execution of $FromLevel(T,k,r,r-1,u)$, all vertices in levels $L_{0},...,L_{r-1}$ are informed. And since each vertex in $L_{r-1}$ broadcasts to the tree leaves, after round 2 all vertices in $L_{r}$ are informed. Thus, at the end of algorithm Alg2, all the vertices of $T$ are informed.
\end{proof}

\begin{lemma}
\label{l12}
The algorithm Alg2 fulfills the edge disjoint constraint.
\end{lemma}
\begin{proof}
Since the rounds are executed sequentially, calls that are executed in different rounds are edge-disjoint. Two calls that are executed by $FromLevel$ or two calls that are executed by $ToLevel$ have been proven to be edge-disjoint in lemmas \ref{l6} and \ref{l8}, respectively. Thus, it is necessary to prove only that two calls, $c_{1}$, $c_{2}$ that are executed in round 2 at the same time unit, but are not both executed by $FromLevel$, are edge-disjoint. There are 4 possible cases:
\begin{enumerate}

\item $c_{1}$ is executed by $FromLevel$, and $c_{2}$ is between a vertex $v \in L_{r-1}$ and a vertex $x \in L_{r}$. $c_{1}$ uses a path from a vertex in $L_{r-1}$ to a vertex in $L_{i}$, for $0\le i \le r-2$, and $c_{2}$ uses only the edge $(v,x)$. Thus, these calls are edge-disjoint.

\item $c_{1}$ is between a vertex $v \in L_{r-1}$, and a vertex $x \in L_{r}$ and $c_{2}$ is between a vertex $z \in L_{r-1}$, $z \neq v$, and a vertex $w \in L_{r}$ ($w \neq z$ because $T$ is a tree). Thus, $c_{1}$ uses only the edge $(v,x)$, and $c_{2}$ uses only the edge $(z,w)$, so that $c_{1}$ and $c_{2}$ are edge-disjoint.

\item $c_{1}$ is between a vertex $v \in L_{r-1}$ and a vertex $x \in L_{r}$, and $c_{2}$ is between two vertices $z,w \in L_{r}$. Thus, $c_{1}$ uses only the edge $(v,x)$, and $c_{2}$ uses two edges, $(z,P(z))$ and $(P(z),w)$. Since $x \neq z$ and $x \neq w$, $c_{1}$ and $c_{2}$ are edge-disjoint (even if $P(z)=v$).

\item $c_{1}$ is between two vertices $v,x \in L_{r}$, and $c_{2}$ is between two vertices $z,w \in L_{r}$, $z,w \not\in \{v,x\}$.
    Therefore, $c_{1}$ uses the edges $(v,P(v))$ and $(P(v),x)$, and $c_{2}$ uses two edges, $(z,P(z))$ and $(P(z),w)$. Thus, $c_{1}$ and $c_{2}$ are edge-disjoint.

\end{enumerate}
\end{proof}

\subsubsection{The total time of Alg2}
The execution of the procedure $ToLevel$ at the first round takes $ \lceil \log_{2}(k^{r-1}+1) \rceil$ time units.

The execution of the procedure $FromLevel$ at the second round takes one unit of time, and the broadcast from the vertices in $L_{r-1}$ to the vertices in $L_{r}$ takes $\lceil \log_{2}(k+1) \rceil$ time units, which is the total time of the second round.

Thus, the total time needed to complete Alg2 is $ \lceil \log_{2}(k^{r-1}+1) \rceil + \lceil \log_{2}(k+1) \rceil \le \log_{2}((k^{r-1}+1)(k+1)) = \log_{2}(k^{r}+k^{r-1}+k+1) $

\subsubsection{The cost of Alg2}
The cost of the first round is derived from (3), where $r-1$ is substituted for $j$.
The cost of the second round is $k^{r-1}(2k-\lceil \log_{2}(k+1) \rceil)$.

Thus, the cost of Alg2 is at most
$$[2-\frac{(k-1)}{k^2}\lceil \log_{2}(k+1)\rceil + \frac{1}{k(k-1)}]n-2(r-1)+\frac{k}{(k-1)^2}+\frac{1}{k}-\frac{\lceil \log_{2}(k+1) \rceil}{k^2}.\eqno(6)$$

This proves (2) of theorem 2. 

\subsection{Alg3}
This algorithm is based on the procedures $ToLevel$ and $FromLevel$. First, the originator broadcasts to $L_{r}$, namely, the tree leaves, and then the originator and the vertices in $L_{r}$ broadcast to all the tree vertices. The execution of $FromLevel$ starts at the last time unit of the execution of $ToLevel$, where an informed vertex transmits an uninformed sibling.

\begin{algorithm}[h!!]
\caption{Alg3(T,k,r,u)}
\label{alg:EDP2}
\begin{algorithmic}[1]
	\STATE 	{ToLevel   $(T,k,r,r,u)$}
	\STATE 	{FromLevel $(T,k,r,r,u)$}
	\end{algorithmic}
\end{algorithm}


\subsubsection{Correctness of Alg3}
\begin{lemma}
\label{l14}
At the end of the execution of the algorithm all vertices in $T$ are informed.
\end{lemma}
\begin{proof}
 According to lemma \ref{l5}, after the execution of $ToLevel(T,k,r,r,u)$, $u$ and all vertices in $L_{r}$ are informed. According to lemma \ref{l7}, after the execution of $FromLevel(T,k,r,r,u)$, all vertices in levels $L_{0},...,L_{r}$ are informed. Thus, at the end of the algorithm Alg3 all vertices in $T$ are informed.
\end{proof}

\begin{lemma}
\label{l15}
The algorithm Alg3 fulfills the edge disjoint constraint.
\end{lemma}
\begin{proof}
Since the procedures $ToLevel$ and $FromLevel$ are executed sequentially, the proof follows directly from lemmas \ref{l6} and \ref{l8}.
\end{proof}

\subsubsection{The total time of Alg3}
Since each vertex is active at all time units left from the time unit he receives the message, the number of informed vertices doubles at each time unit and therefore, the total time of Alg3 is $\lceil \log_{2}n \rceil$.



\subsubsection{The cost of Alg3}




The cost of Alg3 is based upon the cost of the procedures $ToLevel$ and $FromLevel$, where $j=r$.

By substituting $r$ for $j$ in (4) and (5), summing (4) and (5), and using the fact that $n=\frac{k^{r+1}-1}{k-1}$ (and therefore $k^{r} = \frac{n(k-1)+1}{k}$), we obtain:
$$B(u) \le (2 + { \frac{1}{k-1}})n + 2r \lceil \log_2 k^{r} \rceil - 2 \lceil \log_2 (k^{r}+1) \rceil  - 3r- { \frac{r+1}{k-1}}. $$


This proves (3) of theorem 2.


\begin{thebibliography}{20}

\bibitem{ahr}
A. Averbuch, R. Hollander-Shabtai, Y. Roditty. k-Port Line broadcasting in Trees.
{\it Journal of Combinatorial Mathematics and Combinatorial Computing} 77(2011), 125-160.


\bibitem{agr1}
A. Averbuch, I. Gaber, Y. Roditty. Low cost minimum-time line
broadcasting in complete binary trees. {\it Networks} 38(2001),
189-193.

 \bibitem{ars1}
 A. Averbuch, Y. Roditty, B. Shoham. Computation of
 broadcasting multiple messages in a positive weighted tree.
 {\it J. of Combinatorial Mathematics, Combinatorical Computation}
35(2000),161-184.

\bibitem{ars2}
A. Averbuch, Y. Roditty, B. Shoham. Efficient line broadcast in a
$d-$ dimensional grid. {\it Discrete Applied Mathematics} 113
(2001), 129-141.

\bibitem{cfm}
J. Cohen, P. Fraigniaud, M. Mitjana. Polynomial time algorithms for
minimum-time broadcast in trees. {\it Theory of Computing Systems},
35(6), 2002, 641-665.

\bibitem{f}
A. Farley. Minimum-time line broadcast networks. {\it Networks}
10(1980), 59-70 .

 \bibitem{fh}
 A. Farley, S. Hedetniemi.
 Broadcasting in grid graphs. {\it In proceedings of the Ninth SE Conference on Combinatorics, Graph Theory and Computing
 Winnipeg}, 1978, 175-288.

\bibitem{fl}
P. Fraigniaud and E. Lazard. Methods and Problems of Communication in Usual Networks.
{\it Discrete Applied Mathematics} 53 (1994), p. 79-133.

\bibitem{ff}
S. Fujita, A. Farley. Minimum-cost line broadcasting in paths. {\it
Discrete Applied Mathematics} 75(1997), 255-268.

 \bibitem{hhl}
 S.M. Hedetniemi, S.T. Hedetniemi, A.L. Liestman. A survey of
 gossiping and broadcasting in communication networks.
 {\it Networks}  18 (1988), 319-349.

\bibitem{kp}
J.O. Kane, J.G. Peters. Line broadcasting in cycles. {\it Discrete
Applied Mathematics} 83(1998), 207-228.

\bibitem {gp}
G. Peretz Alon. Line Broadcasting in Trees. {\it M.Sc. Thesis, The Open University of Israel}, 2010.

\bibitem {p}
A. Proskurowski. Minimum broadcast trees. {\em IEEE Transactions on
Computers} c-30(1981), 363-366.

 \bibitem{rs}
 Y. Roditty, B. Shoham. On broadcasting multiple messages in a
 $d-$dimensional grid.
 {\it Discrete Applied Mathematics}, 75(1997), 277-284.

 \bibitem{sch}
 P.J. Slater, E.J. Cockayne, S.T. Hedetniemi.
 Information dissemination in trees. {\it SIAM J. Comput.} 10(1981) No. 4, 692-701.

 \bibitem{vs}
 F.L. Van Scoy, J.A. Brooks.
 Broadcasting multiple messages in a grid.
 {\it Discrete Applied Mathematics} 53(1994), 321-336.

 \bibitem{west}
Douglas B. West. Introduction to Graph Theory. {\it Prentice Hall}
(1996).


\end{thebibliography}
\end{document}